\documentclass[pra,twocolumn,showpacs,preprintnumbers,amsmath,amssymb,superscriptaddress]{revtex4-1}

\usepackage{amsmath}
\usepackage{amsfonts}
\usepackage{amsthm}
\usepackage{color}
\usepackage{graphicx}
\usepackage{dcolumn}
\usepackage{bm}
\usepackage{bbm}
\usepackage{times}
\usepackage{changes}
\usepackage{comment}
\usepackage{enumitem}
\usepackage{ulem}
\usepackage[caption=false]{subfig}
\usepackage[pdfstartview=FitH]{hyperref}
\usepackage[figure]{hypcap}
\definecolor{myurlcolor}{rgb}{0,0,0.7}
\definecolor{myrefcolor}{rgb}{0.8,0,0} 
\hypersetup{colorlinks, linkcolor=myrefcolor,
citecolor=myurlcolor, urlcolor=myurlcolor}

\usepackage{blkarray}
\usepackage{mathtools}
\usepackage{multirow,bigdelim}

\usepackage[draft]{fixme}
\usepackage{xspace}

\def\beq{\begin{equation}}
\def\eeq{\end{equation}}
\def\be{\begin{equation}}
\def\ee{\end{equation}}
\def\ben{\begin{eqnarray}}
\def\een{\end{eqnarray}}
\def\beqa{\begin{eqnarray}}
\def\eeqa{\end{eqnarray}}
\def\eea{\end{array}}
\def\bea{\begin{array}}
\newcommand{\bei}{\begin{itemize}}
\newcommand{\eei}{\end{itemize}}
\newcommand{\bee}{\begin{enumerate}}
\newcommand{\eee}{\end{enumerate}}

\def\mt{main text\xspace}
\def\>{\rangle}
\def\<{\langle}
\def\omeas{B_1}
\def\ogentle{B_1^g}
\def\chsh{\beta_{\text{CHSH}}}
\def\chshMax{\beta^{\max}_{\text{CHSH}}}
\def\chain{\beta_{\text{chain}}}
\def\chainMax{\beta^{\max}_{\text{chain}}}
\def\betaMax{\beta^{\max}}
\def\betaMaxcl{\beta^{\max}_{cl}}
\def\dist{{\cal D}}
\def\loc{{\cal L}}
\def\inf{{\cal I}}
\def\relev{w}

\newtheorem{prop}{Proposition}

\newcommand\bovermat[2]{%
  \makebox[0pt][l]{$\smash{\overbrace{\phantom{%
    \begin{matrix}#2\end{matrix}}}^{\text{#1}}}$}#2}

\newcommand{\ket}[1]{| #1 \rangle}
\newcommand{\bra}[1]{\langle #1 |}
\newcommand{\proj}[1]{| #1 \rangle   \langle #1 |}
\def\tr{\textrm{Tr}}

\begin{document}

\title{Measurement uncertainty from no-signaling and non-locality}

\author{Justyna \L{}odyga}
\email{jlo@amu.edu.pl}
\affiliation{Faculty of Physics, Adam Mickiewicz University, 61-614 Pozna\'{n}, Poland}

\author{Waldemar K\l{}obus}
\affiliation{Faculty of Physics, Adam Mickiewicz University, 61-614 Pozna\'{n}, Poland}

\author{Ravishankar Ramanathan}
\affiliation{Institute for Theoretical Physics and Astrophysics,
University of Gda{\'n}sk, 80-952 Gda{\'n}sk, Poland}
\affiliation{National Quantum Information Centre of Gda\'{n}sk, 81-824 Sopot, Poland}
\affiliation{Laboratoire d'Information Quantique, Universit\'{e} Libre de Bruxelles, Belgium}

\author{Andrzej Grudka}
\affiliation{Faculty of Physics, Adam Mickiewicz University, 61-614 Pozna\'{n}, Poland}

\author{Micha\l\ Horodecki}
\affiliation{Institute for Theoretical Physics and Astrophysics,
University of Gda{\'n}sk, 80-952 Gda{\'n}sk, Poland}
\affiliation{National Quantum Information Centre of Gda\'{n}sk, 81-824 Sopot, Poland}

\author{Ryszard Horodecki}
\affiliation{Institute for Theoretical Physics and Astrophysics,
University of Gda{\'n}sk, 80-952 Gda{\'n}sk, Poland}
\affiliation{National Quantum Information Centre of Gda\'{n}sk, 81-824 Sopot, Poland}

\date{\today}

\begin{abstract}
One of the formulations of Heisenberg uncertainty principle, concerning so-called measurement uncertainty, states that the measurement of one observable modifies the statistics of the other. Here, we derive such a measurement uncertainty principle from two comprehensible assumptions: impossibility of instantaneous messaging at a distance (no-signaling), and violation of Bell inequalities (non-locality). The uncertainty is established for a pair of observables of one of two spatially
separated systems that exhibit non-local correlations. To this end, we introduce a gentle form of measurement which acquires
partial information about one of the observables. We then bound disturbance of the remaining observables
by the amount of information gained from the gentle measurement, minus a correction depending on the degree of non-locality.
The obtained quantitative expression resembles the quantum mechanical formulations, yet it is derived without the quantum formalism and complements the known qualitative effect of disturbance implied by non-locality and no-signaling.
\end{abstract}

\maketitle

In recent decades much effort was done to understand quantum mechanics ``from the outside''. Namely, one considers possible constraints
for correlations coming solely from no-signaling principle, and compares them with quantum mechanical constraints. The first observation was already made in the nineties by Popescu and Rohrlich \cite{PopescuRohrlich1994}. They showed that no-signaling constraints are much weaker, and allow for extremely strong correlations that violate the so called Bell-CHSH inequality \cite{Bell1964,CHSH1969} to the maximal possible extent, i.e., achieving maximal algebraic value of the Bell quantity.

On the other hand, much work was done in order to extract features of quantum formalism that are responsible for various non-classical effects, such as quantum computational speedup, reduction of communication complexity, quantum key distribution and expansion or amplification of weak randomness. It turns out that to achieve at least some of those effects, one does not need to employ the full quantum formalism, but just refer to its two features: the impossibility of faster-than-light communication (no-signaling) combined with Bell non-locality.
For example, to obtain secure key distribution, one uses just the no-signaling principle in conjunction with the fact that statistics obtained in distant labs violate Bell inequalities, exhibiting in this way Bell non-locality \cite{BarrettHardyKent2005}. However, such a fundamental rule as the Heisenberg uncertainty principle \cite{Heisenberg1927}, so far treated as a hallmark of quantum mechanics, has not yet been derived from these simple assumptions.

When considering the Heisenberg uncertainty principle, one may think of either of its two faces: the {\it preparation uncertainty principle}, stating that one cannot prepare a system in a state exhibiting peaked statistics for each of two incompatible observables \cite{Robertson1929,WinterWehner2010, Bialynicki2011}, and the {\it measurement uncertainty principle}, stating that by measuring one observable, one necessarily disturbs the statistics of the other observable \cite{Ozawa2003, Ozawa2004PhysLett,Erhart2012,BuschLahtiWerner2013}. Tomamichel and H\"anggi \cite{TomamichelHanggi2013} obtained the former principle from non-locality using the quantum formalism. However, the preparation uncertainty cannot be determined solely from 
no-signaling and non-locality, as it is not exhibited by the Popescu-Rohrlich box \cite{OppenheimWehner2010}. 
The measurement uncertainty principle, on the other hand, does not meet such restriction. It has a closely related formulation as an information gain versus disturbance trade-off \cite{FuchsPerez1996,Dariano2003,Maccone2006,Banaszek2006,BuscemiHayashiHorodecki2008,BuscemiHallOzawaWilde2014} and has become a basis for quantum cryptography \cite{BB84,B92}: a potential eavesdropper by gaining information about the cryptographic key necessarily disturbs the system, which can be noticed by the parties that are to establish the key.
The subject of measurement uncertainty principle in the context of non-locality and no-signaling was touched upon by Oppenheim and Wehner \cite{OppenheimWehner2010} who showed (in a non-quantitative manner) that Bell non-locality implies that a sharp measurement, i.e., the measurement with complete knowledge about the outcome, must cause disturbance.

In this Letter, we derive a {\it quantitative} measurement uncertainty relation, in the form of a trade-off implied by Bell non-locality and no-signaling. To this end, we introduce a notion of gentle measurement as well as a quantitative notion of disturbance, both applicable in 
the operational scenario, where the only objects are statistics of measurements. In particular, we consider a bipartite scenario where Bob, who exhibits non-local correlations with Alice (measured by degree of violation of a chosen Bell inequality), performs consecutive measurements of a pair of his observables. As a result, we find that the very act of his first measurement disturbs the statistics of the second measurement (this happens even if the first measurement is \textit{gentle}, i.e., where he does not acquire full knowledge about the result). Additionally, it appears that the magnitude of such disturbance increases not only with information gain but also with the strength of Bell inequality violation. We subsequently compare our result with its counterpart obtained within the quantum mechanical framework. 

In our findings we use traditional monogamy relations to obtain dynamical-type (or better kinematic-type) relations. The former are static, and state that
if two systems are non-locally correlated, the possible information present in a third system must be limited.
In contrast, we consider a time ordered scenario where a party measures observables one by one, exactly like in the measurement uncertainty principle. 

\textit{Information gain via a gentle measurement.---} 
We start with an initial bipartite system, one system possessed by Alice, the other by Bob.
Alice and Bob can sharply measure their observables $A_x$, $x=1,\ldots, n$ and $B_y$, $y=1,\ldots, m$, respectively, and obtain corresponding outcomes $a$  and $b$.
In addition, for Bob we introduce a gentle measurement responsible for the partial gain of information of one of his observables. Hereafter, without loss of generality, we choose the fixed observable $\omeas$ to be measured gently.
Bob will perform  the gentle measurement before he measures another observable, by coupling his measuring apparatus to the system. Equivalently, we can imagine that a third person --- Grace --- couples
some other system to Bob's one, performs some evolution, and takes away her system.
This results in an overall tripartite system: on two of them Alice and Bob can still measure their sharp observables,
while Grace can measure her single observable that represents gentle measurement of Bob's chosen observable $\omeas$.
In terms of no-signaling boxes, Grace has just one input, which we call $\ogentle$, with corresponding output $b_1^g$.

Formally, let us denote the statistics of the original bipartite system as  $p(a,b|A_x,B_y)$,
and the statistics of the tripartite system as $\tilde p(a,b,b_1^g|A_x,B_y,\ogentle)$,
where $\ogentle$ -- the gentle version of $\omeas$ -- is the only observable available to Grace.
The final bipartite statistics is then given by 
\be
\label{eq:finalstatistics}
\tilde p(a,b|A_x,B_y,\ogentle)=\sum_{b_1^g} \tilde p(a,b,b_1^g|A_x,B_y,\ogentle).
\ee

We shall now require that Grace's observable is indeed a gentle version of Bob's observable, by imposing two conditions (for details see Appendix \ref{secGENTLE}):

\bee[nolistsep,leftmargin=*]
\item  The act of Grace's measurement will not affect the statistics of the sharp observable $\omeas$, conditioned on any input and output of Alice, i.e.,
\be
p(b_1|\omeas,a,A_x)= \tilde  p(b_1|\omeas,\ogentle,a,A_x), \quad \forall {a,x}.
\label{eq:equal}
\ee

\item Grace's output $b_1^g$ will be correlated with Bob's output of measurement of
$\omeas$ (again conditioned on any Alice's input and output) resulting in the following conditional probability distribution
\be
\label{eq:weak}
 \tilde p(b_1^{g}=i|b_1=j,\omeas,\ogentle,a,A_x)=
  \begin{cases}
   \tfrac12+\epsilon & \text{if } i=j,  \\
   \tfrac12-\epsilon & \text{if } i\neq j, \\
  \end{cases}
\ee
where the parameter $\epsilon \in [0,\frac12]$ quantifies the information gain. For $\epsilon=\frac12$, complete information about the observable is acquired,
i.e., the sharp measurement gives the same output as the gentle measurement,
whereas for $\epsilon=0$, the outputs of gentle measurement are completely uncorrelated with the outputs of sharp measurement, hence the information gain is zero.\\
\eee

Let us emphasize that we will not restrict in any way what possible change may happen to the original bipartite box, other than by the above assumptions -- which are imposed just by the very definition of gentle measurement.
The resulting change will follow solely from no-signaling and non-locality.

\textit{Disturbance.---}
Consider first a (not necessarily quantum mechanical) state $\rho$ and a given observable. We want to quantify how much the observable  is disturbed by some other action on the state, that changes it into state $\tilde\rho$; in our case the action is the
gentle measurement of observable $\omeas$. A natural disturbance measure is the statistical distance between the probability distribution $p(b|B_y,a,A_x)$ obtained by measuring the observable $B_y \neq \omeas$
on state $\rho$ (i.e., prior to the gentle measurement) and the distribution $\tilde p(b|B_y,\ogentle,a,A_x)$ obtained by measuring this observable on state $\tilde\rho$ (after the gentle measurement is performed).
While deriving the disturbance from non-locality, we shall not show however that the disturbance holds for some particular state. 
Rather, we prove that disturbance occurs for some of the states produced by Alice.
When Alice chooses an observable $A_x$ and obtains an outcome $a$, a state $\rho_{a,A_x}$ is created at Bob's side. The state changed by gentle measurement is thus given by $\tilde \rho_{a,A_x}$.
Note that
since the gentle measurement is performed on Bob's system, then due to no-signaling we have $p(a|A_x)=\tilde p(a|A_x)$.
For a given choice of Alice's observable $A_x$ and an outcome $a$, the disturbance of the observable $B_y \neq \omeas$ is defined as
\be
D_{a,x}(B_y)= \sum_b |p(b|B_y,a,A_x) - \tilde p(b|B_y,\ogentle,a,A_x)|.
\ee
In this work we consider the average total disturbance, where we sum over all Alice's observables and all Bob's observables apart from $\omeas$ itself, and average over Alice's outcomes
\be
\dist= \sum_{a,x} p(A_x) p(a|A_x) \sum_{y\not = 1}  D_{a,x}(B_y).
\label{eq:dist}
\ee
In Appendix \ref{secDIST} we argue that the change of non-locality necessarily causes disturbance, proving that
for arbitrary Bell inequality (with moduli of coefficients bounded by $1$, w.l.o.g.), the average total disturbance $\dist$ \eqref{eq:dist} satisfies
\be
\label{eq:DB}
n \dist \geq |\beta(p)- \beta(\tilde p)|,
\ee
where $n$ denotes the number of Alice's measurement choices, and $\beta(p), \beta(\tilde p)$ are the values of the Bell quantity evaluated on initial statistics $p(a,b|A_x,B_y)$ and final statistics $\tilde p(a,b|A_x,B_y,\ogentle)$ given by Eq. \eqref{eq:finalstatistics}, respectively.

\textit{Relevance of Bell inequalities for observable.---}
It could happen that a chosen Bell inequality does not cover some of the observables.
For example, in Bell-CHSH inequality for a scenario where Alice and Bob hold $n=2$ and $m=3$ observables, respectively, one of Bob's observables is not included. 
Therefore, such observable does not cause any disturbance.

To quantify the ability of the observable $\omeas$ to disturb the other observable, given a specific Bell inequality, we introduce a new quantity, namely the notion of relevance $\relev(\omeas)$. For simplicity, and due to our convention that the gentle measurement is always performed on a fixed observable $\omeas$, in $\relev(\omeas)$ we neglect the argument  $\omeas$ and define the relevance $\relev$ as 
\be\label{eq:relev}
\relev= \betaMax - \betaMax_1, 
\ee
where $\betaMax$ denotes the maximal value of Bell quantity for no-signaling probabilistic theories and $\betaMax_1$ the maximal value of Bell quantity where the observable $\omeas$ is deterministic. 
The relevance $\relev$  \eqref{eq:relev} measures how far the observable is from being deterministic, i.e., it quantifies its degree of randomness. Therefore, for the increasing value of the relevance $\relev$, we observe stronger disturbance properties of the observable $\omeas$. 

\begin{figure}
\centering
\includegraphics[height=.27\textwidth]{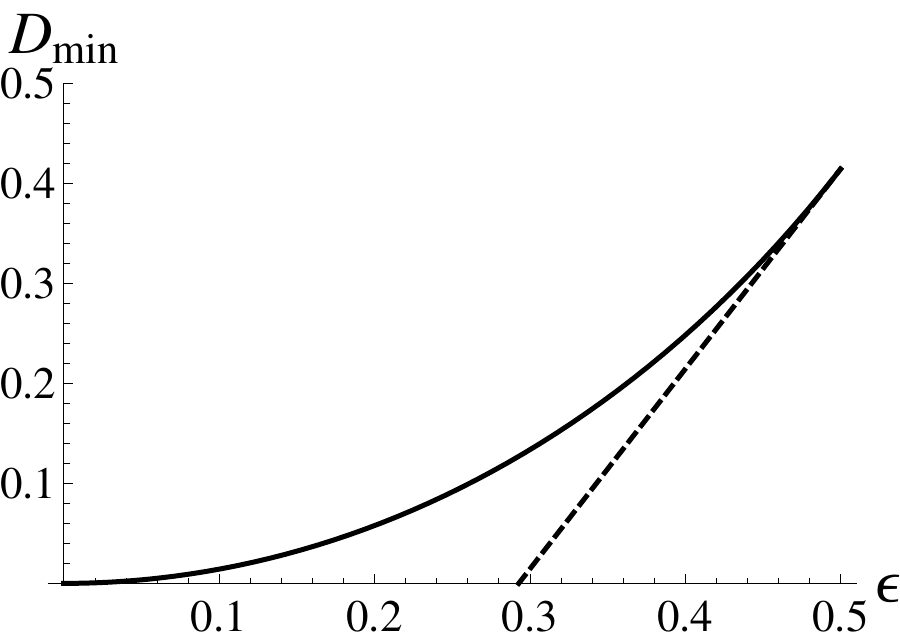}
\caption{Lower bound $\dist_{\min}$ on average total disturbance obtained from: quantum predictions (thick line), no-signaling principle (dashed line) for the case of CHSH inequality, where we choose $\chsh=2 \sqrt{2}$ corresponding to maximally non-local correlations attainable within the framework of quantum mechanics.}
\label{fig1}
\end{figure}

For that reason, the relevance $\relev$ \eqref{eq:relev} determines the strength of a  monogamy relation related to the value $\beta$ of a chosen Bell inequality
\be
\label{eq:mongen}
\beta + \relev \<\ogentle\omeas\> \leq \betaMax,
\ee
where $\<\ogentle\omeas\>$ stands for a correlation function between $\ogentle$ and $\omeas$. In Appendix \ref{secRELEV} we provide a proof for the relation \eqref{eq:mongen}, and show that for the CHSH and chain Bell inequality: $\relev= 2$ for any chosen observable, whereas for so called total function XOR games (a more general class of correlation Bell inequalities with binary outputs): $\relev\geq \min(\betaMax-\betaMaxcl, n)$, where $\betaMaxcl$ denotes the maximal classical value of the Bell quantity.

\textit{Measurement uncertainty principle.---}
We now present our main result, i.e., the trade-off between information gained in the gentle measurement and the disturbance caused by it on the remaining observables. Consider arbitrary Bell inequality, and rescale it so that it can be written as $\beta=\sum_{a,b,x,y} c(a,b,A_x,B_y) p(a,b|A_x,B_y)$, where the coefficients are bounded as $|c(a,b,A_x,B_y)| \leq 1$. For so defined Bell inequality $\beta$, the trade-off is of the following general form
\be
n \dist \geq  \relev \inf-  \loc.
\label{eq:general}
\ee
Here $n$ is the number of Alice's observables; disturbance $\dist$ is given by Eq. \eqref{eq:dist}; the relevance $\relev$ by Eq. \eqref{eq:relev}; the information gain $\inf$ is
defined as $2 \epsilon$ with  $\epsilon$  defined by Eq. \eqref{eq:weak}
(the factor 2 is added for technical reasons, actually $2\epsilon$ has the interpretation of correlation
function between $\ogentle$ and $\omeas$, cf. Appendix \ref{secTRADE});
and finally degree of locality ${\loc} = \betaMax - \beta$ reports how the non-locality of the system departs from maximal possible non-locality, quantified by the violation of a chosen Bell inequality. One can note that whenever the local content $\loc$ vanishes (we are in the extreme point of no-signaling correlations), arbitrarily small information gain causes disturbance.

\textit{Examples.---}
Two exemplary particular trade-offs can be obtained from CHSH inequality, and its generalization -- chain Bell inequality.
The CHSH inequality reads
\be
\chsh=\< A_1 B_1\> + \< A_1 B_2\> + \< A_2 B_1\> - \< A_2 B_2\> \leq 2
\ee
with maximal value $\chshMax=4$.
There are just two observables on either side, thus when Bob gently measures $\omeas$, he disturbs the observable $B_2$, and
the trade-off stands as
\be\label{exCHSH}
{\dist} = D(B_2) \geq 2 \epsilon -\frac12 (4- \chsh),
\ee
where we used Eq. \eqref{eq:general} with $n=2$, $\relev=2$,  $\inf= 2 \epsilon$ and ${\loc} = 4 - \chsh$. For maximally non-local correlations exhibited by so called {\it Popescu-Rohrlich box}, $\chsh=4$,
we simply have $ {\dist} \geq {\inf}$. For non-maximally non-local correlations,  there is some threshold value of $\epsilon$,
for which the inequality \eqref{exCHSH} is non-trivial. For example, at the Tsirelson bound $\chsh=2 \sqrt{2}$, attained for maximal correlations allowed in quantum regime, $\epsilon_{th}=0.293$ as depicted in Fig. \ref{fig1} (dashed line).

\begin{figure}
\centering
\includegraphics[height=.27\textwidth]{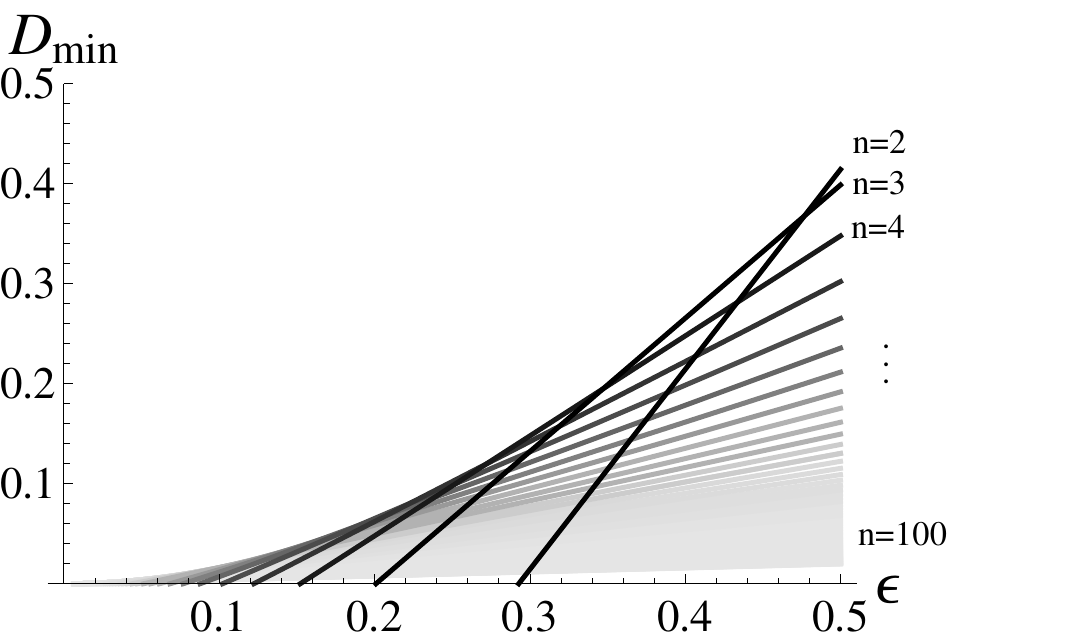}
\caption{Lower bound $\dist_{\min}$ on average total disturbance obtained from non-locality and no-signaling principle for the case of chain Bell inequality, where we choose $\chain=2 n \cos(\frac{\pi}{2 n})$ \cite{Wehner2006} with $n$ denoting the number of observables.}
\label{fig2}
\end{figure}

The chain inequality \cite{BraunsteinCaves1990} is given by
\ben
\label{eq:chain}
\chain &=& \sum_{k=1}^{n-1}  (\< A_k B_k\> + \<A_k B_{k+1}\>) + \< A_n B_n\> - \<A_n B_1\> \nonumber \\
&\leq& 2n-2,
\een
and $\chainMax=2n$ .
Analogous to the CHSH inequality, for the gentle measurement of $\omeas$, we obtain
\be\label{eq:exchain}
{\dist} = \sum_{i\not= 1}D(B_i) \geq \frac{4}{n} \epsilon - \frac{1}{n} (2n- \chain).
\ee
The dependence of disturbance on information gain, as well as on number of observables $n$ is presented
in Fig. \ref{fig2}. Note that the larger the number of observables, the more the threshold $\epsilon_{th}(n)$ moves towards
zero. At the same time, the disturbance goes down as $O(\frac1n)$.

In Appendix \ref{secGCHAIN} we present another example of Bell inequality -- generalized chain inequality -- in a form of total XOR game for which we provide an optimal quantum strategy. It appears that for some range of parameters the obtained disturbance can be even greater (going down with number of observables $n$ as $O(\frac{1}{n^{1/2+\delta}})$ for small $\delta > 0$) than in previous two examples.

\textit{Comparison with quantum uncertainty.---}
We shall now examine, how much the uncertainty imposed solely by non-locality  in no-signaling world is weaker than that implied by non-locality in the quantum mechanical world. To this end, we use quantum monogamy relation for the case of CHSH (for derivation see Appendix \ref{secQMON})
\be
\label{eq:qmonCHSH}
\left(\chsh \right)^2 + 4 |\<\ogentle \omeas\>|^2 \leq 8,
\ee
which together with Eq. \eqref{eq:DB} gives the following trade-off
\be
D_{q}(B)\geq \frac12 \left(\chsh - \sqrt{8- 4(2\epsilon)^2}\right) ,
\ee
with $\<\ogentle \omeas\>=2 \epsilon$. In Fig. \ref{fig1} we illustrate this result for $\chsh=2 \sqrt{2}$ (thick line) and compare with its counterpart in no-signaling world (dashed line). One can notice that the minimal disturbance in the former case is greater than for the latter. Such behavior is expected since no-signaling constraints are in general weaker than quantum mechanical ones \cite{MassarPironio2005}.

\textit{Discussion.---}In this Letter, we have developed a new, more perceptive way of obtaining the measurement uncertainty principle from no-signaling and non-locality. 
In particular, we considered a bipartite scenario where one party chooses to measure one of his observables, whereas the second party first performs a gentle measurement of one observable (gaining only partial information about the outcome) and then, a strong measurement of another observable (where the information gain is maximal). Subsequently, assuming only impossibility of superluminal communication between two parties (i.e., the no-signaling principle) and violation of Bell inequality, we have examined a relation between information gain and disturbance implied by the very act of the gentle measurement. 
Our results for the case of sharp measurement (i.e., $\epsilon=\frac12$) reproduce the extreme case discussed by Oppenheim and Wehner in \cite{OppenheimWehner2010}.

Remarkably, while, as we have shown, non-locality implies measurement uncertainty, the connection between preparation uncertainty 
and non-locality is quite opposite: it has been shown \cite{OppenheimWehner2010} that  preparation uncertainty excludes too strong 
non-locality (cf. \cite{RaviPHorodecki2015}).

Our results indicate that for general probabilistic theories obeying the no-signaling principle, the disturbance 
implied by statistics that can be observed in labs (i.e., the statistics predicted by quantum mechanics)
is trivial until information gain reaches some threshold value of $\epsilon_{th}$. 
This threshold can be shifted towards zero, by considering more observables (as in the case of chain Bell inequality).

Moreover, our trade-off has the following cryptographic interpretation. Alice prepares a bipartite system
and sends one subsystem to Bob. If the latter subsystem is intercepted and measured by an eavesdropper,
then, at the end, Alice and Bob share a disturbed box. For this reason, our results can have potential applications in cryptography based on sending states as in BB84 protocol rather than by performing measurements on shared entangled states of unknown origin.

An open question would be to obtain ultimate envelope describing the trade-off, i.e., to find the largest possible disturbance for a given information gain. In our work, we have found a Bell inequality that leads to disturbance partially greater than for the usual chain inequality, however we only observed it to happen for large number of observables $n$. Therefore, there still remains an open question of how to obtain the optimal Bell inequality implying the largest possible disturbance for a given information gain $\epsilon$, irrespective of the value of $n$ (for the whole range of $n$). Thus, so far our best bound for sought envelope is the one given by chain inequalities.

\begin{acknowledgments}
We would like to thank P. Kurzy{\'n}ski, D. Lasecki and A. W{\'o}jcik for helpful discussions. This work was supported by ERC Advanced Grant QOLAPS and by John Templeton Foundation. The opinions expressed in this publication are those of the authors and do not necessarily reflect the views of the John Templeton Foundation. R.R. acknowledges support from the research project  ``Causality in quantum theory: foundations and applications'' of the Fondation Wiener-Anspach and from the Interuniversity Attraction Poles 5 program of the Belgian Science Policy Office under the grant IAP P7-35 photonics@be.
\end{acknowledgments}

\onecolumngrid

\begin{appendix}
\section{Gentle measurement}\label{secGENTLE}

In this section we provide a more detailed description of the gentle measurement of observable $\omeas$ performed by Bob.
Let us explicitly state the assumptions that the gentle measurement should satisfy. These assumptions are natural and, in particular, are satisfied by a quantum gentle measurement, as we shall see later. Suppose first that we do not measure the gentle observable, but only the sharp one.
The probability distribution of the outcome is denoted by $p(b_1|\omeas,a,A_x)$. Let us also consider a situation where the observable $\omeas$ is first measured gently (denoted as $\ogentle$) and then sharply.
Since, {\it a priori}, the statistics of the latter sharp measurement might be
disturbed by the preceding gentle measurement, we will for a while denote its outcome by $b_1'$. The corresponding resulting probability distribution we denote by $\tilde p(b_1',b_1^g|\omeas,\ogentle,a,A_x)$.
We will now make two assumptions. First, we assume that the marginal probability of outcome $b_1'$ is the same as  that of $b_1$, i.e.,
\be
\tilde p(b_1'=j|\omeas,\ogentle,a,A_x)=p(b_1=j|\omeas,a,A_x),
\ee
where $\tilde p(b_1'|\omeas,\ogentle,a,A_x)=\sum_{b_1^g} \tilde p(b_1',b_1^g|\omeas,\ogentle,a,A_x)$, for any state of the system (recall that various states of Bob's system are prepared by
different choices of Alice's observable  and by different outcomes of her measurements).
Second, we assume that the conditional probability distribution computed from the above mentioned joint probability distribution is given by
\be
\label{gentle_sm}
\tilde p(b_1^{g}=i|b_1'=j,\omeas,\ogentle,a,A_x)=
  \begin{cases}
   \tfrac12+\epsilon & \text{if } i=j,  \\
   \tfrac12-\epsilon & \text{if } i\neq j, \\
  \end{cases}
\ee
which is almost like Eq. \eqref{eq:weak} of the \mt. The only difference is that instead of $b_1$ as in Eq. \eqref{eq:weak},
we have $b_1'$. However, our first assumption implies, in particular, that joint probability distribution  of $b_1'$
with Alice's outcomes is the same as that of $b_1$. Thus for all our purposes, the two random variables are indistinguishable.
Hence we can drop the prime in the above conditions, obtaining Eq. \eqref{eq:weak}.\\

We will now show that quantum measurements satisfy the above assumptions.
To this end, consider a sharp measurement of $\omeas$ described by projection operators
\ben
\label{eq:POVM1}
\hat{P}_0&=& \proj{0}, \\
\label{eq:POVM2}
\hat{P}_1&=& \proj{1},
\een
performed on an arbitrary qubit state
\be
|\Psi\>=\beta |0\> + \sqrt{1-\beta^2}|1\>,
\label{eq:initialBstate}
\ee
with $\beta \in \mathbb{R}, 0 \leq \beta \leq 1$, which leads to the following marginal probability distributions for outcomes $b_1 \in \left\{0,1\right\}$
\ben
\label{original1}
p(b_1=0|\omeas)&=&\beta^2,\\
\label{original2}
p(b_1=1|\omeas)&=&1-\beta^2.
\een
The gentle measurement for $\omeas$ is described by Kraus operators
\ben
\label{EE0}
\hat{E}_0&=&\sqrt{\tfrac12+\epsilon}\proj{0}+\sqrt{\tfrac12-\epsilon}\proj{1}, \\
\label{EE1}
\hat{E}_1&=&\sqrt{\tfrac12-\epsilon}\proj{0}+\sqrt{\tfrac12+\epsilon}\proj{1}.
\een
In order to show that with such definitions of sharp and gentle measurements, the two assumptions mentioned above are satisfied, we consider a procedure where the gentle measurement is followed by the sharp one.

The marginal probability distributions for outcomes $b_1^{g} \in \left\{0,1\right\}$ are given by
\ben
\label{eq:aver1a}
\tilde p(b_1^{g}=0|\ogentle)&=&\tr (\hat{E}_0 \proj{\Psi} \hat{E}^\dagger_0)=\left(\frac12 +\epsilon\right) \beta^2 + \left(\frac12 -\epsilon\right)\left(1-\beta^2\right), \\
\label{eq:aver1b}
\tilde p(b_1^{g}=1|\ogentle)&=&\tr (\hat{E}_1 \proj{\Psi} \hat{E}^\dagger_1)=\left(\frac12 -\epsilon\right) \beta^2 + \left(\frac12 +\epsilon\right)\left(1-\beta^2\right),
\een
where $\ket{\Psi}$ is described in Eq.(\ref{eq:initialBstate}), and $\hat{E}_0$, $\hat{E}_1$ in Eqs.(\ref{EE0})-(\ref{EE1}).

After obtaining the outcomes $b_1^{g}=0$ and $b_1^{g}=1$, the post-measurement states are given by
\ben
\label{eq:poststate1}
\ket{\Psi^{g}_0} = \frac{\hat{E}_0 \ket{\Psi}}{\sqrt{\bra{\Psi} \hat{E}^\dagger_0 \hat{E}_0  \ket{\Psi} }}=\frac{\sqrt{\frac12 + \epsilon} \beta}{\sqrt{\left(\frac12 +\epsilon\right) \beta^2 + \left(\frac12 -\epsilon\right)\left(1-\beta^2\right)}} \ket{0} + \frac{\sqrt{\frac12 - \epsilon} \sqrt{1-\beta^2}}{\sqrt{\left(\frac12 +\epsilon\right) \beta^2 + \left(\frac12 -\epsilon\right)\left(1-\beta^2\right)}} \ket{1},
\een
\ben
\label{eq:poststate2}
\ket{\Psi^{g}_1} = \frac{\hat{E}_1 \ket{\Psi}}{\sqrt{\bra{\Psi} \hat{E}^\dagger_1 \hat{E}_1  \ket{\Psi} }}= \frac{\sqrt{\frac12 - \epsilon} \beta}{\sqrt{\left(\frac12 -\epsilon\right) \beta^2 + \left(\frac12 +\epsilon\right)\left(1-\beta^2\right)}} \ket{0} + \frac{\sqrt{\frac12 + \epsilon} \sqrt{1-\beta^2}}{\sqrt{\left(\frac12 -\epsilon\right) \beta^2 + \left(\frac12 +\epsilon\right)\left(1-\beta^2\right)}} \ket{1}.
\een
The second measurement is thus performed on above post-measurement states, and leads to the following conditional probabilities for the outcome $b_1'=0$
\ben
\label{eq:aver2a}
\tilde p(b_1'=0|b_1^{g}=0,\omeas,\ogentle)&=&\tr (\hat{P}_0 \proj{\Psi^{g}_0} \hat{P}^\dagger_0)=\frac{\left(\frac12 +\epsilon\right) \beta^2}{\left(\frac12 +\epsilon\right) \beta^2 + \left(\frac12 -\epsilon\right)\left(1-\beta^2\right)}, \\
\label{eq:aver2b}
\tilde p(b_1'=0|b_1^{g}=1,\omeas,\ogentle)&=&\tr (\hat{P}_0 \proj{\Psi^{g}_1} \hat{P}^\dagger_0)=\frac{\left(\frac12 -\epsilon\right) \beta^2}{\left(\frac12 -\epsilon\right) \beta^2 + \left(\frac12 +\epsilon\right)\left(1-\beta^2\right)},
\een
where $\hat{P}_0$ is given in Eq.(\ref{eq:POVM1}), and $\ket{\Psi^{g}_0}$ and $\ket{\Psi^{g}_1}$ in Eqs.(\ref{eq:poststate1})-(\ref{eq:poststate2}).\\

Using  $\tilde p(b_1'=j|\omeas,\ogentle)=\sum_{i=0,1} \tilde p(b'_1=j|b_1^{g}=i,\omeas,\ogentle) \tilde p(b_1^{g}=i|\ogentle)$ together with  Eqs.(\ref{eq:aver1a})-(\ref{eq:aver1b}) and Eqs.(\ref{eq:aver2a})-(\ref{eq:aver2b}), we obtain that $\tilde p(b_1'=0|\omeas,\ogentle)=\beta^2$, hence it is equal to $p(b_1=0|\omeas)$ of Eq. \eqref{original1}.
The same reasoning applies to the case of $b_1'=1$.
Therefore, we have showed that our first assumption works for quantum mechanics.

To show that our second assumption holds (i.e., that Eq. \eqref{gentle_sm} holds)
we write
\be
\tilde p(b_1^{g}=i|b_1'=j,\omeas,\ogentle)= \tilde p(b_1'=j|b_1^{g}=i,\omeas,\ogentle)  \frac{\tilde p(b_1^{g}=i|\ogentle)}{\tilde p(b_1'=j|\omeas,\ogentle)}.
\ee
Now, replacing  $\tilde p(b_1'=j|\omeas,\ogentle)$ with  $p(b_1=j|\omeas)$ (since they are equal), and inserting
Eqs.(\ref{eq:aver1a})-(\ref{eq:aver1b}) and Eqs.(\ref{eq:aver2a})-(\ref{eq:aver2b}),  we obtain the required identity.

\section{Disturbance}\label{secDIST}
In this section we examine the relation between change of non-locality and the disturbance caused in the system. In particular, we prove Eq. \eqref{eq:DB} from the main text given in the following form
\be
\label{eq:SDB}
n \dist \geq |\beta(p)- \beta(\tilde p)|,
\ee
where $n$ is the number of Alice's observables and $\beta(p), \beta(\tilde p)$ are the values of the Bell quantity evaluated on initial $p(a,b|A_x,B_y)$ and final statistics $\tilde p(a,b|A_x,B_y,\ogentle)=\sum_{b_1^g} \tilde p(a,b,b_1^g|A_x,B_y,\ogentle)$, respectively. We prove that the change in non-locality, quantified by the change of an arbitrary Bell quantity (with coefficients bounded by $1$), inevitably leads to non-trivial disturbance.\\
\begin{proof}
First, note that
any Bell inequality can be written (up to a constant factor) as
\be
\label{eq:SBell}
\sum_{a,b,x,y} c(a,b,A_x,B_y) p(a,b|A_x,B_y) \leq \beta_{cl},
\ee
where 
\be
\label{eq:Scoeff}
|c(a,b,A_x,B_y)|\leq 1.
\ee
We then have
\ben
&&|\beta(p) - \beta(\tilde p)| = \left|
\sum_{a,b,x,y} c(a,b,A_x,B_y) \left(p(a,b|A_x,B_y) -\tilde p(a,b|A_x,B_y,\ogentle) \right)\right| \leq  \nonumber \\
&&\leq 
\sum_{a,b,x,y} \left|(p(a,b|A_x,B_y) -\tilde p(a,b|A_x,B_y,\ogentle) \right| =  \nonumber \\
&& =
\sum_{a,b,x,y\not=1} \left|(p(a,b|A_x,B_y) -\tilde p(a,b|A_x,B_y,\ogentle) \right| =  \nonumber \\
&& = 
n \sum_{y\not=1}  \left(\sum_{a,x} \frac{1}{n} p(a|A_x) \sum_b |p(b|B_y,a,A_x) - \tilde p(b|B_y,\ogentle,a,A_x)| \right)=
\nonumber  \\
&& =
n \sum_{y\not=1}  \left(\sum_{a,x} \frac{1}{n} p(a|A_x) D_{a,x}(B_y)\right) =n  \dist,
\een
where in the first equality we used Eq. \eqref{eq:SBell}, in the first inequality: Eq.\eqref{eq:Scoeff}, and in the second equality: Eq. \eqref{eq:equal} from the main text, i.e., that $p(b_1|\omeas,a,A_x)= \tilde  p(b_1|\omeas,\ogentle,a,A_x), \forall {a,x}$. In the last equality we assume that all the choices of Alice's observable are equiprobable, i.e., $p(A_x)=\frac{1}{n} \forall x$.
\end{proof}

\section{Relevance of Bell inequalities for observable}\label{secRELEV}
In the main text, for a chosen observable $B_1$ we defined the relevance $\relev(B_1)\equiv \relev$ given by
\be
\label{eq:Srelev}
\relev= \beta^{\max}-\beta^{\max}_{1},
\ee
with $\beta^{\max}$ standing for maximal algebraic value of the Bell quantity and $\beta^{\max}_{1}$ for maximal value of Bell quantity with deterministic observable $B_1$.\\

\subsection{Monogamy relation with relevance $\relev$}\label{secRELEV1}
In this section, we consider the situation where Alice and Bob measure $\vert \mathcal{X} \vert = n$ and $\vert \mathcal{Y}\vert = m$ number of binary observables $A_x$ and $B_y$, respectively. Let us first prove the following monogamy relation (related to some Bell quantity $\beta$) whose strength is determined by the relevance $\relev$ \eqref{eq:Srelev}
 \be
\beta + \relev \<\ogentle\omeas\> \leq \beta^{\max},
\label{eq:Smonogamy}
\ee
where $\<\ogentle\omeas\>$ describes the correlations between observables $\ogentle$ and $\omeas$.

\begin{proof}

Let us consider the tripartite box $\tilde p(a,b,b_1^g|A_x,B_y,\ogentle)$ and convex decompose it as 
\be 
\tilde p(a,b,b_1^g|A_x,B_y,\ogentle) = \sum_i r_i  p_i(a,b|A_x,B_y) \otimes q_i(b_1^g|\ogentle),
\ee 
with $r_i \geq 0, \sum_i r_i = 1$. This can be done owing to the fact that Grace measures a single observable $\ogentle$. By convexity, it is sufficient to restrict the analysis to boxes of the form $p(a,b|A_x,B_y) \otimes q(b_1^g|\ogentle)$.  
Let us further decompose the bipartite box $p(a,b|A_x, B_y)$ shared by Alice and Bob into two types of extremal boxes. The extremal boxes in the two-party scenario for arbitrary number of inputs and binary outputs were classified in \cite{JonesMasanes2005}. From this classification, we see that with probability $p_N$ we have a box with fully random observable $B_1$, and with probability $p_D$, a box where the observable $B_1$ is deterministic. In the first case, the statistics of $\omeas$ is fully correlated with other observables of the Alice-Bob's box producing a fully random output which gives $\<\ogentle\omeas \>_N=0$, whereas in the second case, the statistics of $\omeas$ being uncorrelated with other observables is deterministic which for appropriate choice of $q(b_1^g|\ogentle)$ gives $\<\ogentle\omeas \>_D=1$.  Then $\<\ogentle\omeas\> = p_N \<\ogentle\omeas\>_N + p_D \<\ogentle\omeas\>_D $, where $ \<\ogentle\omeas\>_N=0$ and $ \<\ogentle\omeas\>_D=1$. Therefore, $p_D = \<\ogentle\omeas\>$.
Now, for any Bell quantity $\beta$
\be
\beta \leq p_N \beta^{\max} + p_D \beta^{\max}_{1} = (1- \<\ogentle\omeas\>) \beta^{\max} + \<\ogentle\omeas\> \beta^{\max}_{1} =  \beta^{\max}  - (\beta^{\max} - \beta^{\max}_{1}) \<\ogentle\omeas\>
\ee
and we recover \eqref{eq:Smonogamy} with substitution \eqref{eq:Srelev}.
\end{proof}

\subsection{Examples of relevance $\relev$}\label{secRELEV2}
\begin{enumerate}
\item For total function XOR games with uniform probabilities of inputs, i.e., correlation Bell inequalities of binary outputs with $\pm 1$ coefficients.\\

The relevance $\relev$ is defined in Eq. \eqref{eq:Srelev}. Let us restrict the analysis to extremal boxes \cite{JonesMasanes2005}. In order to obtain $\betaMax_1$ we must consider all extremal boxes with $\omeas$ being deterministic. In general, such boxes can have more than one deterministic observable. Suppose then that the box is defined by having $k_A$ deterministic observables on Alice's side and $k_B$ deterministic observables on Bob's side. In such a case, the matrix of correlators $C=\<A_x B_y \>$, where $x=1,...,n$ and $y=1,...,m$ takes the form\\
\be\label{Scorrmat1}
\begin{matrix}
C
 =
 \begin{bmatrix}
\bovermat{\normalsize{$m$-$k_B$}} {\begin{matrix} \text{ }&\text{ }& \text{ } \\ \text{ }&[\beta_{\text{ns}}]  &\text{ } \\ \text{ } & \text{ } & \text{ } \end{matrix}} & \vline & \hspace{-0.5em} \bovermat{\normalsize{$k_B$}}{\begin{matrix} \text{ } \\ \hspace{0.4em} [0] \\ \text{ } \end{matrix}} \\[0.5em]
\hline
\begin{matrix} \text{ } &  [0] & \text{ } \end{matrix} & \vline & \begin{matrix} [\beta_{cl}] \end{matrix} \\[0.5em]
  \end{bmatrix}
  \begin{aligned}
  &\left.\begin{matrix}
   \vspace{0.5em}
  \\
  \\
  \end{matrix} \right\} %
 n-k_A,\      &\begin{matrix}
  \end{matrix}\\ %
  &\left.\begin{matrix}
   \\ 
  \end{matrix}\right\}%
  k_A,\     \end{aligned}
 \end{matrix}
 \ee

where $[0]$ denotes the zero matrix with respective dimensions, and $[\beta_{\text{ns}}]$ ($[\beta_{cl}]$) the matrix of correlators for the no-signaling (classical) part of the box.
Analyzing the non-zero part of the matrix $C$ \eqref{Scorrmat1}, we conclude that the Bell quantity for such box depends on the number of deterministic observables, such that
\ben
\beta_1 \leq \max \{ (n-k_A)(m-k_B)+ k_Ak_B, \betaMaxcl \}.
\een
Notice that the value $(n-k_A)(m-k_B)+ k_Ak_B$ is maximized only if $k_A =n$ and $k_B =m$, in which case $k_Ak_B=\betaMaxcl$, or if $k_A =0$ and $k_B =1$ where the correlation matrix becomes\\
\be\label{Scorrmat2}
\begin{matrix}
C'
 =
 \begin{bmatrix}
\bovermat{\normalsize{$m$-$1$}} {\begin{matrix} \text{ }&\text{ }& \text{ } \\ \text{ }&\text{ }& \text{ } \\ \text{ }&[\beta_{\text{ns}}]  &\text{ } \\ \text{ } & \text{ } & \text{ } \\ \text{ } & \text{ } & \text{ } \end{matrix}} & \vline & \hspace{-0.5em} \bovermat{\normalsize{$1$}}{\begin{matrix} \text{ } \\ \text{ } \\ \hspace{0.4em} [0] \\ \text{ } \\ \text{ } \end{matrix}} \\[0.5em]
  \end{bmatrix}
  \begin{aligned}
  &\left.\begin{matrix}
   \vspace{0.5em}
  \\
  \\
  \\
  \\
  \end{matrix} \right\} %
 n.\    \end{aligned}
 \end{matrix}
 \ee
Hence, we obtain
\ben
\label{Sbetatotal}
\beta_1 \leq\max \{ n(m-1),\betaMaxcl \}.
\een
Eventually, substituting RHS of Eq. \eqref{Sbetatotal} to the definition of relevance $\relev$ \eqref{eq:Srelev}, we have
\be
\label{eq:Srelevtotal}
\relev_{\text{tot}}  \geq \min(n, \betaMax-\betaMaxcl),
\ee
where we derived the first term in the bracket by taking $\betaMax= n m $. Note that $\relev_{\text{tot}} = n$ for a generic total function XOR game, when the coefficient matrix $C$ is a random Bernoulli matrix, i.e., each entry $C_{ij}$ takes value $\pm 1$ with probability $\frac12$ independent of other entries. This can be seen for example from the bound on $\| \cdot \|_{\infty \rightarrow 1}$ shown in \cite{GT09} which translates to the statement that for such random \textsc{XOR} games, the expected classical value is bounded as
\be 
\betaMaxcl \leq 2( n \sqrt{m} + m \sqrt{n}).
\ee

\item For Bell-CHSH inequality.\\
Directly from the value of relevance $\relev$ obtained for total function XOR games in Eq.  \eqref{eq:Srelevtotal} with the substitution: $n=2$, $\betaMax=4$ and $\betaMaxcl=2$, we obtain
\be
\relev_{\text{CHSH}}=2.
\ee

\item For chain Bell inequality.\\
Since the box with one deterministic observable cannot violate the chain inequality, we obtain $\beta^{\max}_{1}=2n-2$. Therefore, from Eq. \eqref{eq:Srelev} we get
\be
\relev_{\text{chain}}=2
\ee
with the substitution $\betaMax=2n$.

\end{enumerate}

\section{Information gain versus disturbance trade-off}\label{secTRADE}
In this section we prove our main result (Eq. \eqref{eq:general} in the main text)
\be
n \dist \geq  \relev \inf-  \loc,
\label{eq:Sgeneral}
\ee
where $\relev$ is given by Eq. \eqref{eq:Srelev}, $\inf=\<\ogentle\omeas\>$ denotes the information gain, and $\loc=\beta^{\max} - \beta$ the degree of locality.

First, let us show that 
\be
\label{eq:Scorrelator}
\inf \equiv\<\ogentle\omeas\>=2 \epsilon
\ee
\begin{proof}
\be
\<\ogentle\omeas\>=p(b_1^g=b_1)-p(b_1^g\neq b_1)=\frac12+\epsilon-(\frac12-\epsilon)=2\epsilon,
\ee
where in the second equality we used the formula (3) from the main text.
\end{proof}

Now, we can prove our main result \eqref{eq:Sgeneral}. 
\begin{proof}
\be
n \dist \geq \beta(p)- \beta(\tilde p) \geq \beta - \beta^{\max} + \relev 2 \epsilon,
\ee
where in the first inequality we used Eq.  \eqref{eq:SDB} and in the second inequality we used Eq. \eqref{eq:Smonogamy} for $\beta=\beta(\tilde p)$.

Therefore
\be
n \dist \geq \relev 2 \epsilon - (\beta^{\max} - \beta)
\ee
and we obtain Eq. \eqref{eq:Sgeneral} with $\inf= 2 \epsilon$ \eqref{eq:Scorrelator} and $\loc = \beta^{\max} - \beta$.
\end{proof}

\section{Generalized chain inequality}\label{secGCHAIN}
Suppose that Alice and Bob receive inputs $x, y \in [n]$ and output $a, b \in \{0,1\}$. We consider the correlation Bell inequality (partial function $\textsc{XOR}$ game) $\mathcal{I}_{n,k}$ described by the coefficient matrix $C = (t_{y-x})_{x,y=1}^n$ with
\begin{equation}
\label{eq:gen-ch}
t_l =
\begin{cases}
\phantom{-} 1,
& \text{if } \vert l \vert \leq k -1 \; \; \vee \; \; l = k,\\
-1, & \text{if } \vert l \vert \geq n-k+2 \; \; \vee \; \; l = -(n-k+1), \\
0 & \text{else}
\end{cases}
\end{equation}
for a fixed parameter $k \leq n/2$. The coefficient matrix thus has the following banded Toeplitz form

\begin{equation}
C=
\begin{array}{l}
  \ldelim\{{2}{0mm}[$k$] \\ \\ \\ \\ \\ \\\\ \\ \\  \\[0mm]  \ldelim\{{2}{4mm}[$k$] \\ \    \end{array} \\[-1ex]
\begin{bmatrix}
\bovermat{k+1 }{1 & 1  & 1 & 1 &0 &} \ldots & 0 & \bovermat{k-1 }{-1 & -1}\\
1 & 1  & 1 & 1 & 1 & 0 & \ldots & 0 & -1\\
1 & 1 & \ldots   & 1 & 1 & 1 & 0 & \ldots & 0\\
0 & 1 & 1 & \ldots  & 1 & 1 & 1 & 0 & 0\\
\vdots & 0 & 1 & 1 & \vdots & 1 & 1 & 1 & 0 \\
\hdotsfor{9}\cr 0 & \vdots & 0 & 1 & 1 & \vdots & 1 & 1 & 1 \\
0 & \vdots & \vdots & 0 & 1 & 1 & \vdots &1 &1  \\
-1 &0 & \ldots & \ldots & 0 & 1 & 1 & 1 & 1\\
-1 & -1 & 0 & \ldots & \ldots & 0 & 1 & 1 & 1\\
-1 & -1 & -1 & 0 & \ldots & \ldots & 0 & 1 & 1
\end{bmatrix} 
\end{equation}

\begin{prop}
The relevance $w(B_i)$ of observable $B_i$ for the inequality $\mathcal{I}_{n,k}$ given by the coefficient matrix in (\ref{eq:gen-ch}) with parameter $k \leq n/2$ is $w(B_i) = 2k$ for any $i \in [n]$. The no-signaling value of the inequality is given by $\beta_{\text{ns}} = 2 k n$. The quantum value of the inequality is given by 
\begin{equation}
\beta_{q} =  n \csc{\left(\frac{\pi}{2n} \right)} \sin{\left(\frac{k\pi}{n} \right)}.
\end{equation}
For $n$ divisible by $k$, the classical value of the inequality is given by
\begin{equation}
\beta_{cl} = 2 k (n - k).
\end{equation}
\end{prop}
\begin{proof}
Recall that the relevance $w(B_i)$ is defined by $w(B_i) = \beta^{\text{max}} - \beta_{i}^{\text{max}}$ with $\beta_{i}^{\text{max}}$ being the maximum no-signaling value of the Bell quantity when observable $B_i$ is forced to be deterministic. Now, the maximal no-signaling value of the Bell quantity is evidently equal to the maximal algebraic value (the inequality being an $\textsc{XOR}$ game for which there always exists a no-signaling strategy that wins), and is given by 
\begin{equation}
\beta_{\text{ns}} = \beta^{\text{max}} = 2 k n,
\end{equation}
since for every input $x$ of Alice, there are $2k$ inputs $y$ of Bob such that the coefficients $C_{x,y}$ obey $\vert C_{x,y} \vert = 1$.  

Now, we follow an analogous argument to the total function $\textsc{XOR}$ games by setting observable $B_i$ to be deterministic, and considering all the extremal no-signaling boxes from \cite{JonesMasanes2005}. Let $k_A$ denote the number of Alice's observables for which she returns a deterministic output in the extremal no-signaling box and let $k_B$ denote the number of Bob's observables set to be deterministic. For $k_A, k_B \leq 2k$, the value achieved by this no-signaling strategy is given by
\begin{equation}
\label{eq:beta-i-max}
\beta_{i}^{\text{max}} \leq 2k (n - k_A - k_B) + 2 k_A k_B.  
\end{equation}

The other strategy to check is the fully deterministic (classical) strategy. We claim that for $n$ divisible by $k$ 
\begin{equation}
\label{eq:cl-val}
\beta_{cl} = 2k n - 2k^2. 
\end{equation}
This value is achieved when Alice and Bob deterministically output $a, b = 0$ for all $x, y$. 

We will prove Eq. (\ref{eq:cl-val}) by writing the coefficient matrix $C$ as a sum of $k^2$ chain Bell expressions, each with $n/k$ inputs so that the classical value of the individual chain expressions is $2(n/k - 1)$. 
Accordingly, the corresponding chain expressions are given by
\begin{eqnarray}
\sum_{i=0}^{(n/k)-2} A_{j+i k + l-1} \left(B_{j+ik} + B_{j+(i+1)k}\right) + A_{j+n-k+l-1}\left(B_{j+n-k} - B_j \right) \leq 2(n/k-1)  \; \; \forall j \in [k], l \in [k]
\end{eqnarray}
with $A_{n+m} := - A_m$ for all $m \in [k]$. The classical value (\ref{eq:cl-val}) then follows from the sum of the classical value of the chain inequalities, i.e., $(k^2)(2(n/k - 1)) = 2nk-2k^2$. 
Evidently, the optimal value for $w$ is then given from (\ref{eq:beta-i-max}) by $k_A = 1, k_B = 0$ which achieves the value $2kn - 2k$ giving that $w(B_i) = 2k$. 

We now show the optimal quantum strategy for the game. Consider the strategy given by measuring the state
\begin{equation}
| \phi_{+} \rangle = \frac{1}{\sqrt{2}}  \left(|00 \rangle + |11 \rangle \right)
\end{equation}
with observables  
\begin{eqnarray}
A_x &=& \sin{(\theta_x)} \sigma_x + \cos{(\theta_x)} \sigma_z, \nonumber \\
B_y &=& \sin{(\theta_y)} \sigma_x + \cos{(\theta_y)} \sigma_z, 
\end{eqnarray}
where $\sigma_x, \sigma_z$ are the standard Pauli matrices and the measurement angles are given by
\begin{eqnarray}
\theta_x = (x-1) \frac{\pi}{n}, \quad  \theta_y  = (2y-1) \frac{\pi}{2n}.
\end{eqnarray}
This strategy gives the following correlations
\begin{eqnarray}
 \langle A_{x+j} B_{x} \rangle = \cos{\left(\frac{(2j+1)\pi}{2n}\right)}, \quad \langle A_{x} B_{x+j} \rangle = \cos{\left( \frac{(2j-1)\pi}{2n}\right)} \; \; \forall \;\; 0 \leq j \leq n-1.
\end{eqnarray} 
It therefore achieves the value $\beta_q \geq  \sum_{j=1}^{k} 2 n \cos{\left(\frac{(2j-1) \pi}{2n} \right)}$ for the Bell quantity. Let us now show that this strategy is in fact optimal.

To do this, we show that the strategy achieves the upper bound on $\beta_q$ given as $\beta_q \leq n \Vert C \Vert$ \cite{Wehner2006, RaviAugusiakMurta2016, Linden2007}, where $\Vert C \Vert$ denotes the spectral norm, i.e., the maximal singular value of the coefficient matrix $C$. While $C$ given in (\ref{eq:gen-ch}) is a Toeplitz matrix, it is not circulant, but a ``sign-flipped circulant matrix" with each row obtained from the previous row by a shift to the right and a sign change on the corresponding entry. Still, we consider as an ansatz the system of eigenvectors $| \lambda_j \rangle$ with $j  \in \{0, \dots, n-1\}$ with entries
\begin{eqnarray}
\label{eq:eig-vec}
| \lambda_{j} \rangle_{i} = \omega_{j}^{n-i},
\end{eqnarray}
with $\omega_j = \exp{\left(\frac{-i \pi(2j+1)}{n}\right)}$. The corresponding eigenvalues of $C$ are then given by 
\begin{eqnarray}
\label{eq:eig-val}
\lambda_j = \frac{\sum_{i=1}^{k+1} \omega_j^{n-i} - \sum_{i=n-k+2}^{n} \omega_j^{n-i}}{\omega_j^{n-1}}.
\end{eqnarray}
It is readily seen that the eigenvalue equations are satisfied, the $m$-th eigenvalue equation being, for $m \leq k-1$
\begin{eqnarray}
\left(\sum_{i=1}^{k+m} \omega_j^{n-i} - \sum_{i=n-k+m+1}^{n} \omega_j^{n-i} \right) |\lambda_{j} \rangle_m = \lambda_j \omega_j^{n-m} 
\end{eqnarray}
which is satisfied by (\ref{eq:eig-vec}) and (\ref{eq:eig-val}) by applying multiple times the identity $\exp{\left(-i \pi (2j+1)\right)} = -1$. Similarly, for $k \leq m \leq n-k$,
\begin{eqnarray}
\sum_{i=m-k+1}^{k+m} \omega_j^{n-i}  |\lambda_{j} \rangle_m = \lambda_j \omega_j^{n-m}, 
\end{eqnarray}
and for $n-k+1 \leq m \leq n$,
\begin{eqnarray}
\left(\sum_{i=m-k+1}^{n} \omega_j^{n-i} - \sum_{i=1}^{m-n+k} \omega_j^{n-i} \right) |\lambda_{j} \rangle_m = \lambda_j \omega_j^{n-m}. 
\end{eqnarray}
The singular values of $C$ are then given from (\ref{eq:eig-val}) by $\vert \lambda_j \vert$, so that the upper bound $n \Vert C \Vert$ is given after simplification by 
\begin{eqnarray}
\beta_q \leq \sum_{j=1}^{k} 2 n \cos{\left(\frac{(2j-1) \pi}{2n} \right)} = n \csc{\left(\frac{\pi}{2n} \right)} \sin{\left(\frac{k\pi}{n} \right)}.
\end{eqnarray} 
The qubit strategy achieving this bound shows that the strategy is optimal. 
\end{proof}

For the inequality given by (\ref{eq:gen-ch}), the information gain versus disturbance trade-off is given as
\begin{eqnarray}
\label{Sgen_chain_dist}
\dist \geq \frac{4 k \epsilon}{n} - 2 \left(k - \sum_{j=1}^{k} \cos{\left(\frac{(2j-1) \pi}{2n} \right)}\right).
\end{eqnarray}
The second term tends to zero for appropriate choice of $k$. With $\cos{\left(\frac{(2j-1) \pi}{2n} \right)} = 1 - \left(\frac{(2j-1) \pi}{2n}\right)^2 + O\left(\frac{(2j-1)^4}{n^4}\right)$, and $\sum_{j=1}^k (2j-1)^2 = (4k^2-1)k/3$, 
we see that one may choose up to $k = O(n^{1/2-\delta})$ for any $\delta > 0$ such that $n^{2\delta} > \pi^2/(6 \epsilon)$ to get a non-trivial information gain versus disturbance relation, with $\dist = O(n^{-1/2-\delta})$.  

In Fig. \ref{fig3} we compare the obtained trade-off in Eq. \eqref{Sgen_chain_dist} with the trade-off for chain Bell inequality depicted in the \mt in Eq. \eqref{eq:exchain}, and show the case where the former outperforms the latter. To this end, we choose the number of Alice's measurement choices in a range $n=100,...,1000$.

\begin{figure}[h]
\centering
\includegraphics[height=.27\textwidth]{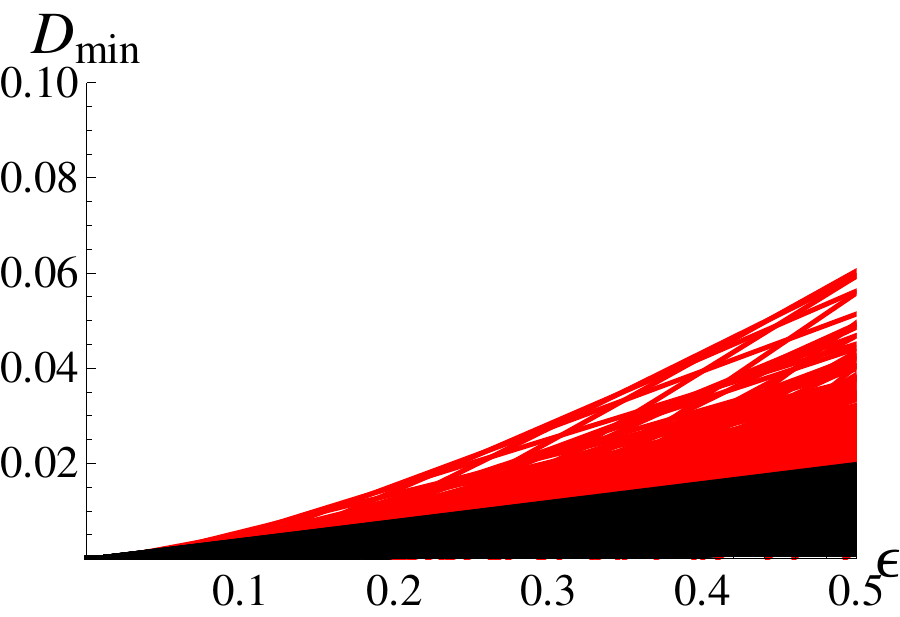}
\caption{The comparison of lower bound $\dist_{\min}$ on average total disturbance implied by chain Bell inequality (red lines): $\dist_{\min} = \frac{4}{n} \epsilon - \frac{1}{n} (2n- \chain)$, where $\chain=2 n \cos(\frac{\pi}{2 n})$, with that implied by generalized chain inequality (black lines): $\dist_{\min} = \frac{4 k \epsilon}{n} - 2 \left(k - \sum_{j=1}^{k} \cos{\left(\frac{(2j-1) \pi}{2n} \right)}\right)$ given in Eq. \eqref{Sgen_chain_dist}, for $n=100,...,1000$.}
\label{fig3}
\end{figure}

\section{Quantum monogamy relation for CHSH inequality}\label{secQMON}

Here, we prove a quantum monogamy relation for the case of CHSH in the following form (Eq. \eqref{eq:qmonCHSH} in the main text)
\be
\label{eq:SqmonCHSH}
\left(\beta_{\text{CHSH}} \right)^2 + 4 |\<\ogentle \omeas\>|^2 \leq 8
\ee
\begin{proof}
To this end, we use the result of \cite{TonerVerstraete2006} that
\be
\label{eq:2chsh}
\left(\beta_{\text{CHSH}}^{AB} \right)^2 +\left(\beta_{\text{CHSH}}^{BC} \right)^2 \leq 8,
\ee
where 
\ben
\label{Schsh1}
\beta_{\text{CHSH}}^{AB} =\< A_1 B_1\> + \< A_1 B_2\> + \< A_2 B_1\> - \< A_2 B_2\>,\\
\label{Schsh2}
\beta_{\text{CHSH}}^{BC}=\< B_1 C_1\> + \< B_1 C_2\> + \< B_2 C_1\> - \< B_2 C_2\>.
\een
Now, let us choose $C_1 = C_2 = \ogentle$. Therefore from \eqref{Schsh2} we get $\beta_{\text{CHSH}}^{BC}=2 |\< B_1 \ogentle\>|$. Substituting this into Eq. \eqref{eq:2chsh}, we obtain \eqref{eq:SqmonCHSH}.
\end{proof}

\end{appendix}

\bibliographystyle{apsrev4-1}
\normalem
\bibliography{Bib_dist_monogamy}

\end{document}